\documentclass[a4paper,english]{lipics-v2016}

\usepackage{placeins}
\usepackage{pgf}
\usepackage{subcaption}
\usepackage{tikz}
\usepackage{todonotes}
\usepackage{tabularx}
\usepackage{hyperref}

\usetikzlibrary{calc,fit,positioning,shapes.geometric}

\title{$K$-Best Solutions of MSO Problems on Tree-Decomposable Graphs}
\author[1]{David Eppstein}
\author[2]{Denis Kurz}
\affil[1]{Computer Science Department, UC Irvine, USA\\
  \texttt{eppstein@uci.edu}}
\affil[2]{Department of Computer Science, TU Dortmund, Germany\\
  \texttt{denis.kurz@tu-dortmund.de}}
\date{}

\DeclareMathOperator\compose{\circ}
\DeclareMathOperator{\card}{\textbf{Card}}
\DeclareMathOperator{\mintwo}{\text{min}_2}
\DeclareMathOperator{\mink}{\text{min}_k}

\newcommand{\lab}{\textit{lab}}
\newcommand{\verts}{\textit{vert}}
\newcommand{\src}{\textit{src}}
\newcommand{\edg}{\textit{edg}}
\newcommand{\typ}{\textit{type}}
\newcommand{\sat}{\textit{\textbf{sat}}}
\newcommand{\bag}{\textit{bag}}

\newcommand{\emptyvec}{\mbox{\boldmath$\emptyset$}}



\begin{document}

\maketitle

\begin{abstract}
We show that, for any graph optimization problem in which the feasible solutions can be expressed by a formula in monadic second-order logic describing sets of vertices or edges and in which the goal is to minimize the sum of the weights in the selected sets, we can find the $k$ best solutions for $n$-vertex graphs of bounded treewidth in time $\mathcal O(n+k\log n)$. In particular, this applies to the problem of finding the $k$ shortest simple paths between given vertices in directed graphs of bounded treewidth, giving an exponential speedup in the per-path cost over previous algorithms.
\end{abstract}

\section{Introduction}

Finding multiple alternative routes between a given pair of vertices in a network, formalized as the \emph{$k$ shortest paths problem}, has a huge number of applications including biological sequence alignment, metabolic pathway reconstruction, hypothesis generation in natural language processing, computer network routing, and vehicle routing~\cite{DBLP:journals/eatcs/Eppstein15}.
When the input network is a directed acyclic graph, it can be solved in constant time per path, after a preprocessing stage whose time is bounded by the time for a single-source shortest path problem~\cite{DBLP:journals/siamcomp/Eppstein98}.
However, for graphs with cycles, this algorithm may find paths with repeated vertices, which are undesirable in many applications. A variant of the problem that disallows these repetitions, the \emph{$k$ shortest simple paths problem}, has also been studied, but for the past 45 years there have been no asymptotic improvements to an algorithm of Yen, which takes quadratic time per path~\cite{10.2307/2629312}. Indeed, a reduction of Vassilevska Williams and Williams suggests that, at least for the case $k=2$, no substantial improvement to this algorithm is possible for arbitrary directed graphs~\cite{DBLP:conf/focs/WilliamsW10}. Even for the easier case of $k$ shortest simple paths in an undirected graph, known algorithms take the time of a single-source shortest path computation per path~\cite{DBLP:journals/networks/KatohIM82}, far from the constant time per path that can be attained when paths are not required to be simple.

This situation suggests studying classes of graphs for which faster $k$-best optimization algorithms are possible. In this paper we provide a first result of this type, showing that the $k$ shortest simple paths can be found in logarithmic time per path (an exponential speedup over the per-path time of previous algorithms) for the graphs of bounded treewidth. Our results are based on  general algorithmic metatheorems that have been developed for these graphs, and in particular on \emph{Courcelle's theorem}, according to which a wide class of decision and optimization problems expressible in the monadic second-order logic of graphs (MSO) can be solved in linear time on graphs of bounded treewidth. MSO is a form of logic in which the variables of a formula represent vertices, edges, sets of vertices, and sets of edges of a graph, one can test set membership and vertex--edge incidence, and variables can be existentially or universally quantified. For instance the property that an edge set $P$ represents a simple path from $s$ to $t$ can be expressed in MSO as a formula
\[
\begin{split}
(\forall v,e,f,g)\Bigl[& \bigl( I(v,e)\wedge e\in P \wedge I(v,f)\wedge f\in P  \wedge I(v,g)\wedge g\in P  \bigr) \Rightarrow \bigl(e=f \vee e=g \vee f=g\bigr)\Bigr]\\
{}\wedge (\exists e)\biggl[ &I(s,e)\wedge e\in P \wedge (\forall f)\Bigl[\bigl(I(s,f)\wedge f\in P)\Rightarrow e=f\bigr)\Bigr]\biggr]\\
{}\wedge (\exists e)\biggl[ &I(t,e)\wedge e\in P \wedge (\forall f)\Bigl[\bigl(I(t,f)\wedge f\in P)\Rightarrow e=f\bigr)\Bigr]\biggr]\\
{}\wedge (\forall S) \Bigl[&
\lnot(\exists v,e) \bigl[ v\in S\wedge e\in P \wedge I(v,e) \bigr] \vee
\lnot(\exists v,e) \bigl[ \lnot(v\in S) \wedge e\in P \wedge I(v,e) \bigr] \vee{}\\
&(\exists v,w,e) \bigl[ v\in S \wedge\lnot(w\in S) \wedge I(v,e) \wedge I(w,e) \bigr]
\Bigr]
\end{split}
\]
(where $s$, $t$, $v$ and $w$ are vertex variables, $e$, $f$, and $g$ are edge variables, $S$ is a vertex-set variable, and $I$ is the vertex--edge incidence predicate).
This formula expresses the constraints that each vertex is incident to at most two edges of $P$,
$s$ and $t$ are each incident to exactly one edge of $P$, and every partition of the vertices that is not crossed by $P$ has $P$ only on one of its two sides.

Courcelle's theorem provides a translation from MSO formulas to tree automata, allowing graphs of bounded treewidth that satisfy the formula to be recognized in linear time by a bottom-up dynamic programming algorithm that executes the tree automaton on a tree-decomposition of the given graph. Arnborg, Lagergren and Seese~\cite{DBLP:journals/jal/ArnborgLS91} and Courcelle and Mosbah~\cite{DBLP:journals/tcs/CourcelleM93} gave extensions of this method that also solve optimization problems for MSO predicates (formulas with one unbound set-variable) that seek the minimum weight vertex set or edge set obeying the given predicate, on a weighted graph of bounded treewidth. For instance, it could be used to find shortest simple paths on bounded-treewidth graphs with negative edges and negative cycles, a problem that is NP-hard on arbitrary graphs.

In this paper, we show that, for any MSO predicate, the $k$ minimum-weight sets satisfying the predicate can be found on graphs of bounded treewidth in logarithmic time per set. In particular, using the formula given above, we can find the $k$ shortest simple paths in logarithmic time per path. Other previously-studied graph optimization problems to which our method applies (and provides an exponential per-solution speedup) include finding the $k$ smallest spanning trees~\cite{DBLP:journals/jacm/EppsteinGIN97}, the $k$ best matchings~\cite{DBLP:journals/dam/ChegireddyH87}, and (with a doubly exponential per-solution speedup) the $k$ best solutions to the traveling salesperson problem~\cite{DBLP:journals/cor/PoortLSV99}. Although the example formula above describes simple paths in undirected graphs, our method applies equally well to directed graphs as long as the undirected graph obtained by forgetting the edge orientations has bounded treewidth.

Our method uses a special tree-decomposition, one that continues to have bounded width but also has logarithmic depth and bounded degree.
We translate the dynamic programming algorithm for finding the minimum weight set satisfying a given MSO predicate into a \emph{fully persistent dynamic graph algorithm} for the same optimization problem, one that can report the minimum weight solution after modifying the given graph by changing the weights of some of its edges or vertices. We apply this method to find the \emph{second-best} (rather than best) solution, and to detect a feature of the graph (a vertex or edge) at which the best and second-best solution differ. By branching on this feature (using the dynamic graph algorithm to delete it in one subproblem and force it to be included in another) we can recursively decompose the original problem into a hierarchy of subproblems whose second-best solutions (together with the global best solution) include all of the $k$ best solutions to the input problem. To find the $k$ best solutions, we perform a best-first search of this hierarchy.

\section{Preliminaries}

In this section, we establish most of our notation and definitions, and review algorithmic components from previous research that we will use to establish our results.

We denote by $[i]$ the set $\{1, \ldots, i\}$ for $i \in \mathbb N$.
The cardinality of a set $M$ is denoted by $|M|$, its power set by $2^M$.
For a function $f: M \to N$, we write $f(M') = \{ f(a) \mid a \in M' \}$ for $M' \subseteq M$.
For an $n$-element sequence $S$ and $i \in [n]$, we denote the $i$-th element of $S$ by $S_i$.

Our algorithms assume a RAM model of computation in which addition and comparison operations on input weights or sums of weights can be performed in constant time per operation.

\subsection{Binary heap of subproblems}
\label{sec:subproblems}

We adopt the following technique for finding the $k$ best solutions to a combinatorial optimization problem. The technique was first used by Gabow for finding the $k$ smallest spanning trees~\cite{DBLP:journals/siamcomp/Gabow77},
and is surveyed by Hamacher and Queyranne~\cite{MR948016}
and Eppstein~\cite{DBLP:journals/eatcs/Eppstein15}.

We assume that the problem to be solved has solutions that can be represented as a set of edges or vertices in a given weighted graph $G$, and that the goal is to minimize the total weight of a solution. We may define a \emph{subproblem} of this problem by forcing certain edges or vertices to be included in the solution and preventing certain other edges or vertices from being included; alternatively, these constraints can be simulated without changing the graph structure by changing the weights of the forced edges or vertices to a large negative number and by changing the weights of the excluded edges or vertices to a large positive number.
We say that a subproblem is \emph{feasible} if there exists a solution to the given problem that is consistent with its constraints.
If $S$ is any subproblem, then we may consider two solutions to $S$, its \emph{best solution} (the one with the minimum weight, subject to the constraints of $S$) and its \emph{second-best solution} (the one that differs from the best solution and otherwise has the minimum possible weight), with ties broken in any consistent way. We say that an edge or vertex is a \emph{pivot feature} if it is present in the best solution but absent in the second-best solution, or vice versa. If a subproblem has only a single solution, then we say that it is \emph{uniquely solvable}.

We then form a binary tree of feasible subproblems, as follows. The root of the tree is the subproblem with no constraints (the one whose solutions are all solutions to the given problem on the whole graph~$G$).
Then, for each subproblem $S$ in the tree that is not uniquely solvable, the two children of $S$ are determined by choosing (arbitrarily) a pivot feature of $S$, constraining that pivot feature to be included in the solutions for one child, and constraining the same pivot feature to be excluded from the solutions for the other child. These two children of $S$ have solution sets that partition the solutions of $S$ into two nonempty subsets, one containing the best solution and the other containing the second-best solution.
A uniquely solvable subproblem in this tree of subproblems forms a leaf, with no children.

Each solution of the given problem on $G$ appears as the second-best solution of exactly one subproblem in this binary tree of subproblems, except for the global best solution which is never the second-best solution of any subproblem. The tree is ordered as a binary min-heap according to the values of the second-best solutions at each subproblem of the tree: the second-best solution of any subproblem is always better than the second-best solutions of its two children. Therefore, the $k$-best solutions of the original problem can be found by outputting the best solution and then performing a best-first search in the tree of subproblems to find the $k-1$ subproblems whose second-best solutions have the smallest values. This best-first search can be performed by evaluating $\mathcal O(k)$ subproblems (finding their best solutions, second-best solutions, and a pivot feature) after using a priority queue of subproblems (prioritized by their second-best solution values) to select each successive subproblem to evaluate in $\mathcal O(\log k)$ time per subproblem. Alternatively, a more complex heap-selection algorithm of Frederickson~\cite{DBLP:journals/iandc/Frederickson93} can be used to find the $k-1$ best subproblems in the tree using (again) $\mathcal O(k)$ evaluations of subproblems but only $\mathcal O(1)$ overhead per subproblem to select each successive subproblem to evaluate.

Using this method, the main remaining task is to show how to find the second-best solution values and pivot features of each subproblem in this tree of subproblems, as efficiently as possible.

\subsection{Path-copying persistence}
\label{sec:persistence}

We will develop a data structure that allows us to add a new constraint to a subproblem, forming one of its child subproblems, and efficiently compute the new second-best solution value of the new child subproblem. However, without additional techniques such a data structure would allow us to follow only a single branch of the tree of subproblems. We use ideas from \emph{persistent data structures}, following Sarnak et al.~\cite{DBLP:journals/jcss/DriscollSST89}, to extend these data structures to ones that let us explore multiple branches of the tree of subproblems concurrently.

An \emph{ephemeral data structure} is one that has only a single version, which is changed by certain \emph{update} operations and accessed but not changed by additional \emph{query} operations. In the corresponding \emph{fully persistent data structure}, each operation takes an additional argument, the \emph{version} of the data structure, and operates on that version. Persistent queries return the result of the query on that version, and do not change it.
Persistent updates create and return a new version of the data structure, in which the given change has been made to the version given as an argument. The previous version is left intact, so that future updates and queries can still be made to it.

\emph{Path copying} is a general technique introduced by Sarnak et al. for converting any tree-based ephemeral data structure (such as a binary heap or binary search tree) into a fully persistent structure for the same problem. In order for it to be applicable, the underlying ephemeral data structure must consist of a tree of nodes, with each node pointing to its children in the tree (along with, possibly other information used as part of updates and queries), but without pointers to parents or other non-child nodes. Every operation in the ephemeral data structure should be performed by starting at the tree root, following child pointers to find additional nodes reachable by paths from the root, and then (in case of a query) collecting information from those nodes or (in case of an update) changing or replacing some or all of the reached nodes.

To make such a data structure fully persistent, we represent each version of the data structure by the root node of its tree.  Persistent queries are handled by exactly the same algorithm as ephemeral queries, starting from the root node representing the desired version of the data structure.
When an ephemeral update would change or replace some subset of the nodes reached during the update, the persistent structure instead creates new nodes for all of these changed or replaced nodes and all of their ancestors, without making any changes to the existing nodes.
In this way, the space and time requirements of each persistent update are proportional to the time for an ephemeral update.

Path-copying persistence was already used in the $k$-shortest paths algorithm of Eppstein~\cite{DBLP:journals/siamcomp/Eppstein98}, as part of the construction of certain persistent heap structures used to represent sets of \emph{detours} in the given graph. Here, we apply the same technique in a different way, to make an ephemeral data structure for second-best solutions persistent. We will associate a version of the second-best solution data structure with each subproblem in the binary tree of subproblems described in the previous section.
Then, when we expand a subproblem (finding its pivot feature and using that feature to define two new child subproblems) the data structure versions of the two child subproblems can be found by applying two different persistent updates to the version of their parent.

\subsection{Shallow tree decompositions}

The data structure to which we will apply the path-copying persistence technique will be based on a tree decomposition of the given graph. However, in order to make the path-copying efficient, we need to use a special kind of tree decomposition, one with low depth.

A \emph{tree decomposition} of a graph $G=(V,E)$ is a pair $(T = (U, F), \bag)$, where $T$ is a tree, $\bag: U \to 2^V$ maps tree nodes to subsets of $V$, and the following conditions are met:
\begin{itemize}
    \item $V = \bigcup \bag(U)$,
    \item for each $e = (v, w) \in E$, there is some $u \in U$ with $\{v, w\} \subseteq \bag(u)$, and
    \item for each $u, u', u''$, if $u'$ is on the path from $u$ to $u''$ in $T$, we have $\bag(u) \cap \bag(u'') \subseteq \bag(u')$.
\end{itemize}
The \emph{width} of a tree decomposition is one less than the size $\max \{ |\bag(u)| \mid u \in U \}$ of its largest bag.
The \emph{treewidth} of a graph $G$ is the smallest $w$ such that $G$ has a tree decomposition of width~$w$.
For any fixed $w$, one can recognize the graphs of treewidth at most $w$ and compute a tree decomposition of optimal width for these graphs, in linear time, using an algorithm of Bodlaender~\cite{DBLP:journals/siamcomp/Bodlaender96}.

We define the \emph{depth} of a tree decomposition to be the longest distance from a leaf to a root node chosen to minimize this distance.
A \emph{shallow tree decomposition} of a graph $G$ of bounded treewidth~$w$ is a tree decomposition with width $\mathcal O(w)$ and depth $\mathcal O(\log |G|)$ whose tree is binary.
Bodlaender~\cite{DBLP:conf/wg/Bodlaender88} showed that a shallow tree decomposition always exists, and that it can be constructed by a PRAM with $\mathcal O(|G|^{3w + 4})$ processors in $\mathcal O(\log |G|)$ time.
The shallow tree decomposition algorithm of Bodlaender and Hagerup~\cite{DBLP:journals/siamcomp/BodlaenderH98} with $\mathcal O((\log |G|)^2)$ running time on an EREW PRAM can be simulated by a RAM in $\mathcal O(|G|)$ time.

\subsection{Hypergraph algebra}

We adopt much of the following notation from Courcelle and Mosbah~\cite{DBLP:journals/tcs/CourcelleM93}.

Let $A$ be a ranked alphabet consisting of \emph{edge labels}, and let $\tau: A \to \mathbb N$ map labels to their \emph{orders}. 
A \emph{hypergraph} $G = (V, E, \lab, \verts, \src)$ \emph{of order $r$} consists of a set of vertices $V$ and a set of hyperedges $E$, an edge labeling function $\lab: E \to A$, a function $\verts: E \to V^*$ that maps edges to node sequences, and a sequence $\src$ of $r$ source nodes.
The \emph{order} of a hyperedge $e \in E$ is the length $\left|\verts(e)\right|$ of its vertex sequence, and must match the order of its label: $\tau(\lab(e)) = \left|\verts(e)\right|$.
A graph of order $r$ is a hypergraph of order $r$ with $\tau(\lab(E)) = \{2\}$, so every hyperedge has order 2.
Hyperedges of a graph may be called \emph{edges}.

We define, for an edge label alphabet $A$, a hypergraph algebra with a possibly infinite set of hypergraph operators and the following finite set of constants.
The constant $\mathbf 0$ denotes the empty hypergraph of order 0.
The constant $\mathbf 1$ denotes the hypergraph of order 1 with a single source vertex and no hyperedges.
For each $a \in A$, the constant $\mathbf a$ denotes the hypergraph of order $\tau(a)$ with node set $\{v_1, \ldots, v_{\tau(a)}\}$, a single hyperedge $e$ with $\lab(e) = a$ and $\verts(e) = (v_i)_{i \in [\tau(a)]}$, and $\src = \verts(e)$.

Let $G$ be a hypergraph of order $r$ and let $G'$ be a hypergraph of order $r'$.
The hypergraph algebra has the following operators.
The $(r + r')$-order hypergraph $G \oplus_{r, r'} G'$ consists of the disjoint union of the vertex and edge sets of $G$ and $G'$, and the concatenation of their source sequences.
For each $i, j \in [r]$, $\theta_{i, j, r}(G)$ is the hypergraph of order $r$ obtained from $G$ by replacing every occurence of $\src_j$ with $\src_i$ in the source sequence of $G$ and in every vertex sequence of a hyperedge of $G$.
For a mapping $\alpha: [p] \to [r]$, $\sigma_{\alpha}(G)$ is the hypergraph of order $p$ obtained from $G$ by replacing its $r$-element source sequence $\src$ with the $p$-element sequence $\src'$, with $\src_i' = \src_{\alpha(i)}$.
We defined an infinite number of operators, and we are able to generate any hypergraph with them.
For each family of hypergraphs $L$ that only contains hypergraphs of bounded treewidth and bounded order, there is a finite subset of the above operators that generates a superset of $L$ \cite{DBLP:journals/jacm/ArnborgCPS93,Bauderon1987,DBLP:books/el/leeuwen90/Courcelle90}.
We denote by $\mathcal G_w$ a finite hypergraph algebra as above that generates all $r$-order hypergraphs over a fixed label set of treewidth at most $w$ and $r < R$ for some fixed but arbitrary $R$.

Let $A$ be an alphabet of edge labels, and let $G = (V, E, \lab, \verts, \src)$ be a hypergraph over $A$.
A formula in \emph{counting monadic second-order logic} (\emph{CMS formula}) is a formula in monadic second-order logic, extended by pradicates $\card_{m, p}(X)$ for $m, p \in \mathbb N$, with $X \models \card_{m, p}(X)$ iff $|X| \equiv m \mod p$, and by incidence predicates $\edg_a(e, v_1, \ldots, v_{\tau(a)})$ for $a \in A$, with $(G, e, v_1, \ldots, v_{\tau(a)}) \models \edg_a(e, v_1, \ldots, v_{\tau(a)})$ iff $\lab(e) = a$ and $\verts(e) = (v_1, \ldots, v_{\tau(a)})$.
The class $\Phi^{h, q}_{A, R}(\mathcal W)$ consists of all CMS formulas for hypergraphs of order at most $R$ over edge labels $A$ on variables $\mathcal W$ whose depth of nested quantification is at most $h$ and $p < q$ for all subformulas of the form $\card_{m, p}(X)$.
Since $h$, $q$, $A$ and $\mathcal W$ are fixed in most contexts, we use the short form $\Phi_R = \Phi^{h, k}_{A, R}(\mathcal W)$.
We require variable alphabets $\mathcal W$ that include constants for all the source nodes of the hypergraphs, i.e., $\{ \mathbf s_i \mid i \in [R] \} \subset \mathcal W$ and $(G, v) \models (v = \mathbf s_i)$ iff $v = \src_i$.
Every other variable $X$ in $\mathcal W$ is assumed to be either a vertex set variable, denoted as $\typ(X) = V$, or a hyperedge set variable, denoted as $\typ(X) = E$.
This can be enforced by only considering formulas that include subformulas $\forall x: x \in X \Rightarrow x \in V$ if X is supposed to model a set of vertices, and $\forall x: x \in X \Rightarrow x \in E$ otherwise.
Note that variables representing single elements can be emulated by a subformula $\exists x: \forall y: x \in X \land (y \in X \Rightarrow x = y)$.

Our algorithms work on parse trees of hypergraphs with respect to some fixed hypergraph algebra $\mathcal G_w$.
A \emph{parse tree} $T = (U, F)$ is a directed rooted tree, with edges being directed away from the root $r$.
Leaves of $T$ are associated with a constant of $\mathcal G_w$; inner nodes are associated with an operator of $\mathcal G_w$.
The hypergraph represented by $T$ is the hypergraph constant associated with $r$ if $r$ is a leaf.
Otherwise, $T$ represents the hypergraph obtained by applying the operator associated with $r$ to the hypergraphs represented by the parse subtrees rooted in the children of $r$.
For $u \in U$, the hypergraph represented by the parse subtree of $T$ rooted in $u$ is denoted by $G(u)$.
Sometimes it is convenient to have exactly two child nodes for each inner node of the parse tree.
We can think of $\sigma_\alpha$ operators as having a second operand that is always the empty hypergraph.
Likewise, $\theta_{i, j, r}$ operators can be viewed as having a second operand that is always the one-vertex hypergraph.
These modified operators can even be derived from the original ones.
Instead of $\sigma_\alpha$, we can use $\sigma_\alpha' = \sigma_\alpha \compose \oplus_{0, r}(\mathbf 0)$.
The composition $\theta_{i, j, r}' = \theta_{i, j, r} \compose \sigma_\alpha \compose \theta_{i, r + 1, r + 1} \compose \oplus_{1, r}(\mathbf 1)$ with $\sigma: [r] \to [r + 1]$, $\sigma(i) = i$, can replace $\theta_{i, j, r}$.
All operators of this derived algebra are binary.
We call a respective parse tree a \emph{full parse tree}.
The \emph{proper} child of a $\theta$ or $\sigma$ node is the one that does not represent the one-vertex or empty hypergraph, respectively.

Let $G$ be a hypergraph, and $T = (U, F)$ a parse tree of $G$.
We consider combinatorial problems on $G$ that can be characterized by a CMS formula $\varphi \in \Phi_R$.
Let $n$ be the number of free variables of $\varphi$, and let $X = (X_i)_{i \in [n]}$ be the free variables themselves.
For example, we need $n = 1$ edge set to describe a (simple) path, or $n = c - 1$ node sets to describe a $c$-coloring of a graph.
An \emph{assignment} maps every $X_i$ to a subset of $\typ(X_i)$.
A \emph{satisfying assignment} is an assignment $f$ such that $(f(X), G) \models \varphi$.
A \emph{solution} is the sequence $(f(X_i))_{i \in [n]}$ for some assignment $f$.
Two solutions $S$, $S'$, are considered distinct if there is a $j \in [n]$ with $S_j \ne S_j'$.
Note that for $S_1 \ne S_2$, $(S_1, S_2)$ is also distinct from $(S_2, S_1)$ regardless of $\varphi$, even if $\varphi$ is a formula on two free variables, and symmetric on these free variables.
For a parse tree node $u \in U$, we denote by $S(u)$ the solution that is obtained from $S$ by removing every graph feature from every set $S_i$ that is not present in $G(u)$.
A solution is \emph{feasible} if the corresponding assignment is satisfying.
We denote by $\sat(G, \varphi)$ the set of feasible solutions, i.e., $Y \in \sat(G, \varphi) \Leftrightarrow (Y, G) \models \varphi$.

Sets of solutions, and $\sat(G, \varphi)$ in particular, are sets of $n$-tuples of sets.
The disjoint union of two sets $M$ and $N$ (which is undefined if $M \cap N \ne \emptyset$) is written as $M \sqcup N$.
Let $X$, $Y$ be solutions, and  $A$, $B$ sets of solutions.
We say that $X$ and $Y$ \emph{interfere} if there are $i, j \in [n]$ such that $X_i \cap Y_j \ne \emptyset$.
By extension, $A$ and $B$ interfere if some $X \in A$ interferes with some $Y \in B$.
If $A$ and $B$ do not interfere, $A \uplus B$ denotes the set of all combinations of each $X \in A$ and $Y \in B$, i.e., $$A \uplus B = \{ (X_i \sqcup Y_i)_{i \in [n]} \mid X \in A, Y \in B \}.$$

The above operators are not defined for every combination of operands.
A \emph{semi-homomorphism} from $\langle \mathcal S, \uplus, \sqcup, \emptyvec, \emptyset \rangle$ to some evaluation structure $\langle \mathcal R, \oplus, \otimes, \mathbf 0, \mathbf 1 \rangle$ is a mapping $f: \mathcal S \to \mathcal R$ that acts like a homomorphism where applicable, i.e., $f(\emptyvec) = \mathbf 0$, $f(\emptyset) = \mathbf 1$, $f(A \uplus B) = f(A) \oplus f(B)$ if $A$ and $B$ do not interfere, and $f(A \sqcup B) = f(A) \otimes f(B)$ if $A \cap B = \emptyset$.

Values of solution sets are expressed in terms of \emph{evaluation structures}.
An evaluation structure is an algebra $\langle \mathcal R, \oplus, \otimes, \mathbf0, \mathbf1 \rangle$ such that $\langle \mathcal R, \oplus, \mathbf 0 \rangle$ and $\langle \mathcal R, \otimes, \mathbf1 \rangle$ are monoids.
An \emph{evaluation} $v$ is a function that maps hypergraphs generated by $\mathcal G_w$ to an (ordered) evaluation structure $\mathcal R$.
An evaluation $v$ is an \emph{MS-evaluation} if there exists a semi-homomorphism $h$ and a CMS formula $\varphi \in \Phi_R$ such that $v(G) = h(\sat(G, \varphi))$ for every hypergraph $G$.

A \emph{linear CMS extremum problem} $P$ consists of an $r$-order hypergraph $G$, a formula $\varphi = \varphi(P) \in \Phi_R$ with $n$ free variables $X$ that characterizes feasible solutions, and a sequence $c$ of $n$ cost functions with $c_i: \typ(X_i) \to \mathbb R$.
This can be generalized to other totally-ordered sets than $\mathbb{R}$ but we will not need this generalization.
The value $c(Y)$ of a solution $Y$ for $P$ is defined as $\sum_{i \in [n]} \sum_{y \in Y_i} c_i(y)$.
An optimal solution of $P$ is a feasible solution $Y^*$ such that for each feasible solution $Y'$ of $P$, we have $c(Y^*) \le c(Y')$.
Our task is to find the value of an optimal feasible solution of $P$, or equivalently, to compute the value of $v(G) = \min \{ c(Y) \mid Y \in \sat(G, \varphi) \}$.
This evaluation can be expressed as $v(G) = h(\sat(G, \varphi))$ with $h(A) = \min \{ c(Y) \mid Y \in A \}$, and is thus an MS-evaluation.
Note that $h$ is a semi-homomorphism, and the corresponding evaluation structure is $\langle \mathbb R \cup \{ \infty \}, +, \min, 0, \infty \rangle$.
The class of all linear CMS extremum problems is called LinCMS.
Let $m = \left|\sat(G, \varphi)\right|$, and $\sat(G, \varphi) = \{ Y^1, \ldots, Y^m \}$.
Let $\Pi = \Pi(G, \varphi) \subseteq S_n([m])$ be the set of permutations such that $(c(Y^{\pi(i)}))_{i \in [m]}$ is nondecreasing for $\pi \in \Pi$.
For a problem $P \in$ LinCMS, \emph{$k$-val$(P)$} is again a problem that gets the same input as $P$ itself, plus some $k \in \mathbb N$.
It asks for the sequence $(c(Y^{\pi(i)}))_{i \in [k']}$ for some $\pi \in \Pi$, where $k' = \min \{ k, m \}$.
The problem \emph{$k$-sol$(P)$} gets the same input as $k$-val$(P)$, but asks for $(Y^{\pi(i)})_{i \in [k']}$ for any $\pi \in \Pi$.
Note that, depending on $\pi$, different outputs for $k$-sol$(P)$ are possible.
However, there is only one valid output for $k$-val$(P)$.
Also note that 1-val$(P)$ is equivalent to $P$ itself.
For the sake of simplicity, we assume $k \ge m$ for the remainder of this article.

\section{The second-best solution}
\label{sec:twoval}

In the following section, we describe how to solve 2-sol$(P)$ in linear time for each $P$ in LinCMS.
In a later section we generalize these results to an arbitrary (constant) number of solutions, and to LOGSPACE and PRAM models of computation.

Courcelle and Mosbah~\cite{DBLP:journals/tcs/CourcelleM93} showed that for a hypergraph $G$ generated by $\mathcal G_w$, and a formula $\varphi \in \Phi_R$, $\sat(G, \varphi)$ can be computed by a bottom-up tree automaton on a parse tree of $G$.
We give a version of their theorem for full parse trees, where every inner node has exactly two child nodes.

\begin{lemma}[Courcelle, Mosbah\cite{DBLP:journals/tcs/CourcelleM93}]
\label{thm:cm}
Let $\varphi \in \Phi_R$, and let $T$ be a full parse tree rooted in $r$ that represents a hypergraph $G$.
If $r$ is not a leaf, it has two child nodes $u_1$, $u_2$ representing $G_1 = G(u_1)$, $G_2 = G(u_2)$, respectively, and the following holds:
\begin{align}
\label{form:subforms}
\sat(G, \varphi) = \biguplus \left\{ \sat(G_1, \psi^k_1) \sqcup \sat(G_2, \psi^k_2) \mid k \in [l] \right\},
\end{align}
where $l$ only depends on the hypergraph operator associated with $r$, and $\psi^k_1$, $\psi^k_2$ are again in $\Phi_R$ for each $k \in [l]$.
\end{lemma}

In the situation of \autoref{thm:cm}, we call $\psi^k_i$ a \emph{child formula} of $\varphi$ with respect to the corresponding operator.
For $k \in [l]$, we call $(\psi^k_1, \psi^k_2)$ a \emph{fitting pair of child formulas}.
A solution $S \in \sat(G, \varphi)$ has a unique fitting pair $(\psi_1, \psi_2)$ of child formulas with $S \in (\sat(G_1, \psi_1) \sqcup \sat(G_2, \psi_2))$, which follows from the fact that all unions in \autoref{form:subforms} are disjoint.
All child formulas $\psi$ with respect to a $\theta_{i, j, m}$ node $u$ can be chosen such that $\src_i \notin S_k$ for any $S \in \sat(G(u'), \psi)$, where $u'$ is the proper child of $u$.
Further, Courcelle and Mosbah demonstrated that every MS-evaluation $v$ can be computed by a similar tree automaton.
The running time required to compute $v$ on the entire hypergraph $G$ is $\mathcal O(|G| \cdot \mu)$, where $\mu$ is the time required to compute $f(A) \oplus f(B)$ and $f(A) \otimes f(B)$ for valid combinations $A$, $B$.
In the uniform cost model, the operators of the evaluation structure $\langle \mathbb R \cup \{ \infty \}, +, \min, 0, \infty \rangle$, addition and selecting the smaller of two real numbers, require $\mathcal O(1)$ time.
Problems in LinCMS can therefore be solved in linear time.
Applying the semi-homomorphism $h$ with $v(G) = h(\sat(G, \varphi))$ to \autoref{form:subforms} yields
\begin{align}
\label{form:subeval}
v(G) = h(\sat(G, \varphi)) = \min \left\{ h(\sat(G_1, \psi^k_1)) + h(\sat(G_2, \psi^k_2)) \mid k \in [l] \right\}.
\end{align}
A different linear-time approach for this special case had been proposed earlier by Arnborg, Lagergren and Seese~\cite{DBLP:journals/jal/ArnborgLS91}.
The basic algorithm of Courcelle and Mosbah computes $h(\sat(G(u), \psi))$ for every parse tree node $u$ in a bottom-up manner, and every $\psi \in \Phi_R$.
Conceptually, we perform a depth-first search on $T$, starting at its root $r$.
Every time we finish a node $v$, we \emph{evaluate} it, i.e., we compute the evaluation $h(\sat(G(v), \phi))$ for every formula $\psi \in \Phi_R$ based on the child formulas of $\psi$ according to \autoref{form:subeval}.
Since the number of formulas and the number of child formulas per formula are fixed, the overall running time is linear in the size of $T$.

Courcelle and Mosbah also propose an improved algorithm, called the \emph{CM algorithm} in this article, that determines in a top-down preprocessing phase the set of formulas that are reachable from $\varphi$ at $r$ via the child formula relation.
We call these the \emph{relevant formulas} of a parse tree node $u$.
Further, we say that $u$ \emph{introduces} a hypergraph feature $x$ if there is a feasible solution for relevant formula $\psi$ of $u$ that contains $x$, but no feasible solution of relevant formulas at the children of $u$ contains $x$.
Hypergraph features can only be introduced by leaf nodes and $\theta$ nodes, because all solutions for $\sigma$ and $\oplus$ nodes are unions of solutions for their child nodes.
A node associated with a $\theta_{i, j, r}$ operation can only introduce the source vertex $\src_i$.
In full parse trees, these $\theta$ nodes have an extra child $v$ that represents a hypergraph consisting only of $\src_i$.
Therefore, all hypergraph features are introduced by leaf nodes in full parse trees.
Without loss of generality, exactly one parse tree node $u(x)$ introduces each hypergraph feature $x$.

Let $u$ be a parse tree node, $\psi_1$, $\psi_2$ two formulas relevant for $u$.
The CM algorithm only computes $h(G(u), \psi_1)$ and $h(G(u), \psi_2)$, which correspond to the values of optimal solutions on $G(u)$ for the problems characterized by $\psi_1$ and $\psi_2$, respectively.
We do not have any information about the solutions themselves besides their values.
In particular, we do not know for any hypergraph feature if it appears in one of those solutions.
Since we do not require optimal solutions to be unique, we also do not know if they represent the same solution even in the case $h(\sat(G(u), \psi_1)) = h(\sat(G(u), \psi_2))$.

In the next section, we need a way to tell if the optimal solutions for $\psi_1$ and $\psi_2$ are the same by only looking at their evaluations.
For this purpose, we establish for each parse tree node $u$ a mapping from the relevant formulas of $u$ to solution IDs in $[|\Phi_R|]$ that has the following discriminating property.
Each relevant formula $\psi$ of $u$ is mapped to an optimal solution for $\psi$ on $G(u)$ such that two formulas are assigned the same solution if and only if they are assigned the same solution ID.
Solution IDs can be computed along with solution values while keeping the linear time bound:

\begin{lemma}
\label{lemma:equality}
Given all solution IDs for all child nodes of a parse tree node $u$, solution IDs for $u$ can be computed in constant time.
\end{lemma}
\begin{proof}
If $u$ is a leaf, let $M_u$ be the set of possible solutions in $G(u)$.
Let $\pi: [|M_u|] \to M_u$ be a permutation of these solutions such that their values are nondecreasing, i.e., $c(\pi(i)) \le c(\pi(i + 1))$ for any $i \in [|M_u| - 1]$.
For each formula $\psi$ relevant for $u$, we assign the smallest $i$ to $h(\sat(G(u), \psi))$ with $c(\pi(i)) = h(\sat(G(u), \psi))$ as the temporary ID for $\psi$.
If $u$ is an inner node,  let $v_1$, $v_2$ be the child nodes of $u$.
At least one fitting pair of child formulas $\rho$ have $h(\sat(G(u), \psi)) = h(\sat(G(v_1), \rho_1) + h(\sat(G(v_2), \rho_2))$.
We choose one of these pairs arbitrarily.
Let $a$, $b$ be the solution IDs associated with $\rho_1$ at $v_1$ and with $\rho_2$ at $v_2$, respectively.
We assign to $\psi$ the temporary ID $a \cdot |\Phi_R| + b$.

By induction, temporary IDs have the discriminating property, but there might be temporary IDs larger than $|\Phi_R|$.
There are at most $|\Phi_R|$ distinct temporary IDs in use, since $u$ cannot have more than $\Phi_R$ many relevant formulas.
We can therefore remap temporary IDs to the range $[|\Phi_R|]$ using an injective compression function.
\end{proof}

The compression function can be stored with $u$, and a mapping from solution ID to solution at leaf nodes, requiring only constant space.
Even if $\sat(G(u), \psi)$ contains multiple optimal solution, it is now possible to refer to \emph{the optimal solution}, which is the one defined recursively in terms of matching solution IDs.
This particular optimal solution can be found by a simple depth-first search based algorithm in linear time, starting at $\varphi$ at the root of the parse tree.
For each $(u, \psi)$, we find a fitting pair of child formulas $\rho_1$, $\rho_2$ for child nodes $v_1$, $v_2$ of $u$, respectively, such that the assigned solution IDs match, process $(v_1, \rho_1)$ and $(v_2, \rho_2)$ independently, and output the associated subsolution at leaf nodes.

To solve 2-val$(P)$, we have to adapt the evaluation structure.
Instead of values in $\mathbb R \cup \{ \infty \}$, we use pairs $(x, y) \in \mathcal R = (\mathbb R \cup \{ \infty \})^2$, where $x$ represents the value of an optimal solution, and $y$ represents the value of a second-best solution.
We define two new binary operators $+_2$ and $\mintwo$ over $\mathcal R$, with $(x_1, y_1) +_2 (x_2, y_2) = (x_1 + x_2, \min(x_1 + y_2, x_1 + y_1))$ and $\mintwo((x_1, y_1), (x_2, y_2)) = (a, b)$, where $a$, $b$ are the smallest and second-smallest element of the multiset $\{ x_1, y_1, x_2, y_2 \}$, respectively.

\begin{lemma}
\label{lemma:twoval}
Let $P$ be a LinCMS problem characterized by the formula $\varphi \in \Phi_R$ and cost functions $c$.
Given a parse tree $T$ of a hypergraph $G$, the CM algorithm solves 2-val$(P)$ in linear time when used in conjunction with the evaluation structure $\langle \mathcal R, +_2, \mintwo, (0, 0), (\infty, \infty) \rangle$.
\end{lemma}
\begin{proof}
For each leaf $u$ of $T$, we solve 2-val$(P)$ on $G(u)$ directly.
Let $S^1 \in \sat(G, \varphi)$ be optimal.
The mapping $v(G) = (c(S_1), \min(c(\sat(G, \varphi) \setminus \{ S^1 \}))$ can be written as $h(\sat(G, \varphi))$ with $h(A \uplus B) = h(A) \mintwo h(B)$ and $h(A \sqcup B) = h(A) +_2 h(B)$.
The operators $\mintwo$ and $+_2$ can be evaluated in constant time, resulting in linear total time using the CM algorithm.
\end{proof}

The concept of solution IDs can be trivially generalized to the new evaluation structure.
This enables us to refer to \emph{the} optimal and second-best solution, and to reconstruct these two solutions in linear time.

\begin{corollary}
Let $P$ be a LinCMS problem, and let $w \in \mathbb N$ be fixed.
Given a parse tree $T$ of a hypergraph $G$ with bounded treewidth, we can solve 2-sol$(P)$ on $G$ in time $\mathcal O(|G|)$.
\end{corollary}

\section{Dynamizing the second-best solution}

In this section, we introduce the evaluation tree data structure that stores intermediate results of the algorithm from the previous section.
The data structure allows for a trivial query for the values of the two best solutions in constant time, and a query for the solutions themselves in linear time.
We demonstrate how to perform a query for a pivot feature, and how to update an evaluation tree to match the two subproblems with respect to this pivot feature as described in \autoref{sec:subproblems}.
Both operations start in the root of the given tree, enabling us to use the path persistence technique described in \autoref{sec:persistence}.

An \emph{evaluation tree} is a tree with the same structure as the parse tree.
We store with every node of the evaluation tree the result of all evaluations of the CM algorithm as described in \autoref{sec:twoval}, as well as all solution ID information.
In addition, we also store the solution sets $\sat(G(u), \psi)$ themselves for each leaf node $u$.
We first describe the query and update procedures for problems characterized by CMS formulas with exactly one free variable.
For this purpose, we identify solutions $(S_1)$ with their first set $S_1$.
For the sake of simplicity, we also identify nodes of the parse tree with their twins in the evaluation tree of the current subproblem.
We assume that the hypergraph $G$ is represented as a full parse tree $T$.

Let $\varphi \in \Phi_R$ be a CMS formula that characterizes a subproblem $P$, and let $x$ be a pivot feature of $P$.
If $x$ is a source vertex of $G$, the two subproblems of $P$ can again be characterized by formulas in $\Phi_R$, namely $\varphi \land (x \in X_1)$ and $\varphi \land (x \notin X_1)$, respectively.
However, if $x$ is a hyperedge or a vertex that is not a source vertex of $G$, there are no such formulas in general.
Therefore, not all subproblems can be characterized by a CMS formula, and we have to generalize the mapping $\sat$ to cover those subproblems.
We define the set $\sat_P(G(u), \psi)$ to contain all solutions in $\sat(G(u), \psi)$ that satisfy the constraints imposed by the binary subproblem tree.
Note that $h(\sat_P(G, \varphi))$ is the solution of 2-val$(P)$.

Let $S^1$ and $S^2$ be the optimal and second-best solution of $P$, respectively.
To find a pivot feature $x$ for $P$, we maintain a current parse tree node $u_j$, two formulas $\psi^1_j$ and $\psi^2_j$, and the invariants $S^1(u_j) \in \sat_P(G(u_j), \psi^1_j)$, $S^1(u_j) \in \sat_P(G(u_j), \psi^2_j)$, and $S^1(u_j) \ne S^2(u_j)$.
In each iteration, we choose $u_{j + 1}$ to be a child node of $u_j$ in the parse tree.
Initially, $u_1$ is the root of $T$, and $\psi^1_1 = \psi^2_1 = \varphi$.
The invariant holds trivially because of $h(\sat_P(G(u_1), \varphi)) = h(\sat_P(G, \varphi))$, and because the optimal solution differs from the second-best one by definition.

If $u_j$ is a leaf, it is possible to construct the fixed-size solutions $S^1(u)$ and $S^2(u)$ explicitly.
If the invariant holds, these solutions are distinct, and we can determine a pivot feature in constant time.
Otherwise, $u_j$ is an inner node with child nodes $v_1$ and $v_2$.
Let $\rho^1$ ($\rho^2$) be a fitting pair of child formulas for $\psi^1_j$ ($\psi^2_j$) for which solution IDs match.
If the invariant holds for $j$, the solutions $S^1(u_j)$ and $S^2(u_j)$ differ, their subsolutions at $v_1$ ($S^1(v_1)$ and $S^2(v_1)$) or those at $v_2$ have to differ as well.
We choose $u_{j + 1}$, $\psi^1_{j + 1}$ and $\psi^2_{j + 1}$ accordingly.
The invariant for $j + 1$ then holds by construction.

\begin{lemma}
Given an evaluation tree of depth $d$ for a subproblem $P$, the above algorithm finds a pivot feature for $P$ in time $\mathcal O(d)$.
\end{lemma}
\begin{proof}
In each iteration, $u_j$ has depth $j$ in the parse tree, so the algorithm has to terminate after $d$ iterations.
Each iteration can be performed in constant time.
In the last iteration, we choose an element in the symmetric difference of $S^1(u_j)$ and $S^2(u_j)$, which is guaranteed to be a pivot feature if the invariant holds.
We already showed by induction that the invariant stays intact.
\end{proof}

Next, we describe the update process to transform the evaluation tree for $P$ into the evaluation tree for one of its two subproblems.
The process is symmetric for both subproblems, so we only describe it for the subproblem $P'$ that requires feasible solutions to contain the pivot feature.

First, we execute the query algorithm to find a pivot feature $x$, and push each node it visits to a stack.
The pivot query algorithm can only terminate at leaf nodes, so the last node $u$ that is pushed to the stack is a leaf.
Recall that for each leaf node $u$ and each relevant formula $\psi$, we store $\sat_{P}(G(u), \psi)$ explicitly.
We enumerate this solution set and remove every solution that does not contain $x$ to obtain $\mathcal S(\psi) = \sat_{P'}(G(u), \psi)$.
To re-evaluate $\psi$, we apply $h$ to $\mathcal S(\psi)$.

Any other node $u$ on the stack is an inner node of the full parse tree, and an ancestor of $u(x)$.
Re-evaluation works the same as the original evaluation in the CM algorithm, by evaluating \autoref{form:subeval} with operators $\mintwo$, $+_2$ instead of $\min$, $+$ as in \autoref{lemma:twoval}.

\begin{lemma}
Given an evaluation tree of depth $d$ for a subproblem $P$, the above algorithm updates the evaluation tree to match a subproblem $P'$ of $P$ in time $\mathcal O(d)$.
\end{lemma}
\begin{proof}
After re-evaluating all formulas at the leaf node $u(x)$, all solutions at $u(x)$ contain $x$.
Let $u$ be an inner node with a child node $v$ such that all solutions at $v$ contain $x$.
Since all solutions at $u$ are unions of one solution at $v$ and one other solution, all solutions at $u$ also contain $x$.
Every ancestor of $u(x)$ is re-evaluated after $u(x)$ itself, so every solution at the root of the evaluation tree of $P'$ contains $x$ by induction.

Filtering solutions of fixed size from a solution set of fixed size, using a criterion that can be checked in constant time, takes again constant time.
Only a fixed number of solution sets need to be filtered.
At inner nodes, we simulate the CM algorithm, which takes constant time per node.
We process the same number nodes as the pivot query algorithm, which is bounded by $d$.
\end{proof}

Let $S$ be the optimal or second-best solution of $P$, whichever remains feasible in $P'$.
We do not require optimal solutions for subproblems to be unique, so we might find two optimal solutions for $P'$ that both differ from $S$.
Alternatively, $S$ might turn up as the second-best solution for $P'$.
Although these results would be valid outcomes of the update procedure, the $k$-best algorithm based on the binary subproblem tree as described in \autoref{sec:subproblems} would still break.
Once we lose track of $S$ by choosing two other solutions as the optimal and second-best solution for $P'$, we cannot detect whether we find $S$ in a subproblem of $P'$ again, leading to (the value of) $S$ being output twice.
That is because solution IDs are not suitable to check two solutions with respect to different evaluation trees for equality efficiently.
Therefore, we need to make sure that $S$ is the optimal solution of $P'$ as encoded by the solution IDs for $P$ and $P'$.
Fortunately, the update procedure can trivially enforce this property.

Now consider a problem that is characterized by $n$ free variables, with $n > 1$.
We may consider $n$ a constant, as it only depends on the problem at hand.
Recall that two solutions $S$, $S'$ are considered distinct if there is a discriminating index $i$ with $S_i \ne S_i'$.
The pivot query algorithm has to take this into account.
For inner nodes, it solely relies on solution IDs to choose a successor of the current parse tree node.
For leaf nodes, there is still a fixed number of feasible solutions for a fixed-size hypergraph.
We can apply the same arguments as above to once more obtain running time $\mathcal O(d)$.
Similarly, we only have to adapt the processing of nodes for the update algorithm.
For $n = 1$, we always imposed the new subproblem constraint on the first and only set of a solution, because the discriminating index was always the same.
Now, we have to take into account the discriminating index as determined during the pivot query phase.
Using the same arguments, updates can still be performed in time $\mathcal O(d)$.
We can now state our main result.

\begin{theorem}
Let $P$ be a LinCMS problem.
Given a graph $G$ with bounded treewidth and a number $k$, computing the values of the $k$ best solutions for $P$ on $G$ requires $\mathcal O(|G| + k \log |G|)$ time and space.
\end{theorem}
\begin{proof}
In linear time, we compute a shallow tree decomposition, transform it into a full parse tree $T$ of depth $d \in \mathcal O(\log |G|)$, and apply the CM algorithm to $T$ according to \autoref{lemma:twoval}.
With the results, we initialize a binary subproblem tree.
We perform a best-first search to find the $k$ best subproblems with respect to their second-best solution as described in \autoref{sec:subproblems}, which requires us to solve $\mathcal O(k)$ subproblems.
Using the persistent tree technique from \autoref{sec:persistence} in conjunction with the update operation above, we need $\mathcal O(d)$ additional time and space per subproblem.
\end{proof}

\section{Fixed numbers of solutions}
We generalize here our algorithm for the second best solution to any fixed number of solutions, and to low-space and parallel complexity classes. We note that similar results could be obtained by expressing the $k$-best solution problem (for a constant $k$) as a formula in MSO with $k$ set variables, all constrained to be solutions, to be different from each other, and with minimum total weight; however, this would give a significantly worse dependence on $k$ than the solution we describe. It would also require additional complication to recover the differences between the solutions (needed for the algorithm in the previous section for non-constant values of~$k$).

Now let $k \in \mathbb N$ be fixed.
The operators $\mink$ and $+_k$ can be defined analogously to the operators $\mintwo$ and $+_2$, respectively.
Evaluating the terms $\mink(x, y)$ and $x +_k y$ for $k$-tuples $x, y \in \mathbb R^k$ still requires constant time.
Hence, the CM algorithm in conjunction with the evaluation structure $\langle (\mathbb R \cup \infty)^k, +_k, \mink, (0, \ldots, 0), (\infty, \ldots, \infty) \rangle$ solves $k$-val$(P)$ for any problem $P$ in LinCMS.
Let $T$ be a parse tree with depth $d$, and let $\le_\text{post}$ be a postordering of $T$.
To solve $k$-val$(P)$, we propose to evaluate the nodes of $T$ according to $\le_\text{post}$.
As soon as a parse tree node has been evaluated, the evaluations of its child node can be dropped.
This way, the number of nodes we store evaluations for never exceeds $d + 1$.

\begin{theorem}
Let $P$ be a LinCMS problem, and let $k, w \in \mathbb N$ be fixed.
Given a graph $G$ with treewidth $w$, the problem $k$-val$(P)$ on $G$ can be solved using logarithmic memory space.
\end{theorem}
\begin{proof}
Elberfeld et al.~\cite{DBLP:conf/focs/ElberfeldJT10} demonstrated that a shallow tree decomposition $T$ can be computed with logarithmic memory space.
Using a depth-first search from an arbitrarily chosen root node, we process the bags of $T$ according to its postordering.
A bag is processed by replacing it with its fixed-size portion of the parse tree, which is then evaluated bottom-up.
The evaluations of all parse tree nodes corresponding to a bag require constant memory space, and we store $\mathcal O(\log |G|)$ of them at a time.
\end{proof}

Further, the CM algorithm can be parallelized as follows.

\begin{theorem}
Let $P$ be a LinCMS problem, and let $k, w \in \mathbb N$ be fixed.
In the EREW PRAM model, given a shallow tree decomposition $T$ of graph $G$ with treewidth $w$, the problem $k$-val$(P)$ on $G$ can be solved in time $\mathcal O(\log |G|)$ by $\mathcal O(|G|)$ processors.
\end{theorem}
\begin{proof}
We allocate one processor $p(u)$ for each node $u$ of $T$, which is responsible for computing the portion of the parse tree corresponding to its bag, and for evaluating all nodes of that portion of the parse tree.
The processor $p(u)$ has to wait until all processors $p(v)$ of child nodes $v$ of $u$ have finished.
The processor of the root node of $T$ therefore has to wait $\mathcal O(\log |G|)$ time.
Only processor $p(u)$ writes solutions for node $u$, and only the parent $u'$ of $u$ reads them.
The $p(u')$ idles until $p(u)$ finishes computation, resulting in exclusive read/write access at any time.
\end{proof}

Finally, using the algorithm of Bodlaender~\cite{DBLP:conf/wg/Bodlaender88} on $\mathcal O(|G|^{3w + 4})$ processors to compute a shallow tree decomposition, we obtain the following.

\begin{corollary}
Let $P$ be a LinCMS problem, and let $k, w \in \mathbb N$ be fixed.
In the CRCW PRAM model, given a graph $G$ with treewidth $w$, the problem $k$-val$(P)$ on $G$ can be solved in time $\mathcal O(\log |G|)$ by $\mathcal O(|G|^{3w + 4})$ processors.
\end{corollary}

\subsection*{Acknowledgements}

This research was performed in part during a visit of Denis Kurz to UC Irvine, supported by DFG GRK 1855 (DOTS). The research of David Eppstein was supported in part by NSF grants CCF-1228639, CCF-1618301, and CCF-1616248.

\bibliography{literature.bib}

\begin{thebibliography}{10}

\bibitem{DBLP:journals/jacm/ArnborgCPS93}
Stefan Arnborg, Bruno Courcelle, Andrzej Proskurowski, and Detlef Seese.
\newblock {An algebraic theory of graph reduction}.
\newblock {\em J. ACM}, 40(5):1134{--}1164, 1993.
\newblock \href {http://dx.doi.org/10.1145/174147.169807}
  {\path{doi:10.1145/174147.169807}}.

\bibitem{DBLP:journals/jal/ArnborgLS91}
Stefan Arnborg, Jens Lagergren, and Detlef Seese.
\newblock {Easy problems for tree-decomposable graphs}.
\newblock {\em J. Algorithms}, 12(2):308{--}340, 1991.
\newblock \href {http://dx.doi.org/10.1016/0196-6774(91)90006-K}
  {\path{doi:10.1016/0196-6774(91)90006-K}}.

\bibitem{Bauderon1987}
Michel Bauderon and Bruno Courcelle.
\newblock {Graph expressions and graph rewritings}.
\newblock {\em Math. Syst. Theory}, 20(1):83{--}127, 1987.
\newblock \href {http://dx.doi.org/10.1007/BF01692060}
  {\path{doi:10.1007/BF01692060}}.

\bibitem{DBLP:conf/wg/Bodlaender88}
Hans~L. Bodlaender.
\newblock {NC-algorithms for graphs with small treewidth}.
\newblock In Jan van Leeuwen, editor, {\em Graph-Theoretic Concepts in Computer
  Science, 14th International Workshop, WG '88, Amsterdam, The Netherlands,
  June 15-17, 1988, Proceedings}, volume 344 of {\em Lecture Notes in Computer
  Science}, pages 1{--}10. Springer, 1989.
\newblock \href {http://dx.doi.org/10.1007/3-540-50728-0_32}
  {\path{doi:10.1007/3-540-50728-0_32}}.

\bibitem{DBLP:journals/siamcomp/Bodlaender96}
Hans~L. Bodlaender.
\newblock {A linear-time algorithm for finding tree-decompositions of small
  treewidth}.
\newblock {\em SIAM J. Comput.}, 25(6):1305{--}1317, 1996.
\newblock \href {http://dx.doi.org/10.1137/S0097539793251219}
  {\path{doi:10.1137/S0097539793251219}}.

\bibitem{DBLP:journals/siamcomp/BodlaenderH98}
Hans~L. Bodlaender and Torben Hagerup.
\newblock Parallel algorithms with optimal speedup for bounded treewidth.
\newblock {\em {SIAM} J. Comput.}, 27(6):1725--1746, 1998.
\newblock \href {http://dx.doi.org/10.1137/S0097539795289859}
  {\path{doi:10.1137/S0097539795289859}}.

\bibitem{DBLP:journals/dam/ChegireddyH87}
Chandra~R. Chegireddy and Horst~W. Hamacher.
\newblock {Algorithms for finding $K$-best perfect matchings}.
\newblock {\em Discrete Applied Math.}, 18(2):155{--}165, 1987.
\newblock \href {http://dx.doi.org/10.1016/0166-218X(87)90017-5}
  {\path{doi:10.1016/0166-218X(87)90017-5}}.

\bibitem{DBLP:books/el/leeuwen90/Courcelle90}
Bruno Courcelle.
\newblock {Graph rewriting: An algebraic and logic approach}.
\newblock In {\em Handbook of Theoretical Computer Science, Volume B: Formal
  Models and Sematics (B)}, pages 193{--}242. 1990.

\bibitem{DBLP:journals/tcs/CourcelleM93}
Bruno Courcelle and Mohamed Mosbah.
\newblock {Monadic second-order evaluations on tree-decomposable graphs}.
\newblock {\em Theor. Comput. Sci.}, 109(1{\&}2):49{--}82, 1993.
\newblock \href {http://dx.doi.org/10.1016/0304-3975(93)90064-Z}
  {\path{doi:10.1016/0304-3975(93)90064-Z}}.

\bibitem{DBLP:journals/jcss/DriscollSST89}
James~R. Driscoll, Neil Sarnak, Daniel~Dominic Sleator, and Robert~Endre
  Tarjan.
\newblock {Making data structures persistent}.
\newblock {\em J. Comput. Syst. Sci.}, 38(1):86{--}124, 1989.
\newblock \href {http://dx.doi.org/10.1016/0022-0000(89)90034-2}
  {\path{doi:10.1016/0022-0000(89)90034-2}}.

\bibitem{DBLP:conf/focs/ElberfeldJT10}
Michael Elberfeld, Andreas Jakoby, and Till Tantau.
\newblock Logspace versions of the theorems of bodlaender and courcelle.
\newblock In {\em 51th Annual {IEEE} Symposium on Foundations of Computer
  Science, {FOCS} 2010, October 23-26, 2010, Las Vegas, Nevada, {USA}}, pages
  143--152. {IEEE} Computer Society, 2010.
\newblock \href {http://dx.doi.org/10.1109/FOCS.2010.21}
  {\path{doi:10.1109/FOCS.2010.21}}.

\bibitem{DBLP:journals/siamcomp/Eppstein98}
David Eppstein.
\newblock {Finding the $k$ shortest paths}.
\newblock {\em SIAM J. Comput.}, 28(2):652{--}673, 1998.
\newblock \href {http://dx.doi.org/10.1137/S0097539795290477}
  {\path{doi:10.1137/S0097539795290477}}.

\bibitem{DBLP:journals/eatcs/Eppstein15}
David Eppstein.
\newblock {$K$-best enumeration}.
\newblock {\em Bull. EATCS}, 115, 2015.
\newblock URL: \url{http://eatcs.org/beatcs/index.php/beatcs/article/view/322}.

\bibitem{DBLP:journals/jacm/EppsteinGIN97}
David Eppstein, Zvi Galil, Giuseppe~F. Italiano, and Amnon Nissenzweig.
\newblock {Sparsification{---}a technique for speeding up dynamic graph
  algorithms}.
\newblock {\em J. ACM}, 44(5):669{--}696, 1997.
\newblock \href {http://dx.doi.org/10.1145/265910.265914}
  {\path{doi:10.1145/265910.265914}}.

\bibitem{DBLP:journals/iandc/Frederickson93}
Greg~N. Frederickson.
\newblock {An optimal algorithm for selection in a min-heap}.
\newblock {\em Inf. Comput.}, 104(2):197{--}214, 1993.
\newblock \href {http://dx.doi.org/10.1006/inco.1993.1030}
  {\path{doi:10.1006/inco.1993.1030}}.

\bibitem{DBLP:journals/siamcomp/Gabow77}
Harold~N. Gabow.
\newblock {Two algorithms for generating weighted spanning trees in order}.
\newblock {\em SIAM J. Comput.}, 6(1):139{--}150, 1977.
\newblock \href {http://dx.doi.org/10.1137/0206011}
  {\path{doi:10.1137/0206011}}.

\bibitem{MR948016}
H.~W. Hamacher and M.~Queyranne.
\newblock {$K$ best solutions to combinatorial optimization problems}.
\newblock {\em Ann. Oper. Res.}, 4(1-4):123{--}143, 1985.
\newblock \href {http://dx.doi.org/10.1007/BF02022039}
  {\path{doi:10.1007/BF02022039}}.

\bibitem{DBLP:journals/networks/KatohIM82}
Naoki Katoh, Toshihide Ibaraki, and Hisashi Mine.
\newblock {An efficient algorithm for $K$ shortest simple paths}.
\newblock {\em Networks}, 12(4):411{--}427, 1982.
\newblock \href {http://dx.doi.org/10.1002/net.3230120406}
  {\path{doi:10.1002/net.3230120406}}.

\bibitem{DBLP:journals/cor/PoortLSV99}
Edo~S. van~der Poort, Marek Libura, Gerard Sierksma, and Jack A.~A. van~der
  Veen.
\newblock {Solving the $k$-best traveling salesman problem}.
\newblock {\em Computers {\&} OR}, 26(4):409{--}425, 1999.
\newblock \href {http://dx.doi.org/10.1016/S0305-0548(98)00070-7}
  {\path{doi:10.1016/S0305-0548(98)00070-7}}.

\bibitem{DBLP:conf/focs/WilliamsW10}
Virginia~Vassilevska Williams and Ryan Williams.
\newblock {Subcubic equivalences between path, matrix and triangle problems}.
\newblock In {\em 51th Annual IEEE Symposium on Foundations of Computer
  Science, FOCS 2010, October 23-26, 2010, Las Vegas, Nevada, USA}, pages
  645{--}654. IEEE Computer Society, 2010.
\newblock \href {http://dx.doi.org/10.1109/FOCS.2010.67}
  {\path{doi:10.1109/FOCS.2010.67}}.

\bibitem{10.2307/2629312}
Jin~Y. Yen.
\newblock {Finding the $K$ shortest loopless paths in a network}.
\newblock {\em Manag. Sci.}, 17(11):712{--}716, 1971.
\newblock URL: \url{http://www.jstor.org/stable/2629312}.

\end{thebibliography}

\end{document}